\documentclass[a4paper]{ltxdoc}
\usepackage{hwstyle}
\usepackage{authblk}

\title{Multilinear Cryptography using Nilpotent Groups}

\author[1]{Delaram Kahrobaei}
\author[2]{Antonio Tortora}
\author[3]{Maria Tota}

\affil[1]{University of York, New York University; e-mail: dk2572@nyu.edu}
\affil[2]{Dipartimento di Matematica e Fisica, Universit\`a della Campania ``Luigi Vanvitelli'', Caserta, Italy; e-mail: antonio.tortora@unicampania.it}
\affil[3]{Dipartimento di Matematica, Universit\`a di Salerno, Fisciano (SA), Italy; e-mail: mtota@unisa.it}

\date{ }

\begin{document}
\maketitle

\begin{abstract}
In this paper we generalize the definition of a multilinear map to arbitrary groups and develop a novel idea of multilinear cryptosystem using nilpotent group identities.\\
 

\end{abstract}

\section{Introduction}
In recent years multilinear maps have attracted attention in cryptography community. The idea has been first proposed by Boneh and Silverberg \cite{BS}.
For $n>2$ the existence of $n$-linear maps is still an open question. One of the main applications of multilinear maps is their use for indistinguishability obfuscation. For example in \cite{LT} Lin and Tessaro proved that trilinear maps are sufficient for the purpose of achieving indistinguishability obfuscation. Recently, Huang \cite{H} constructed cryptographic trilinear maps that involve simple, non-ordinary abelian varieties over finite fields.

Group-based cryptography has some new direction to offer to answer this question. A bilinear cryptosystem using the discrete logarithm problem in matrices coming from a linear representation of a group of nilpotency class $2$ has been proposed in \cite{MS}.

In this paper, we propose multilinear cryptosystems using identities in nilpotent groups, in which the security is based on the discrete logarithm problem.

\section{Multilinear Maps in Cryptography}

Let $n$ be a positive integer. For cyclic groups $G$ and $G_T$ of prime order $p$, a map $e:G^n \rightarrow G_T$ is said to be a (symmetric) $n$-linear map (or a multilinear map) if for any $a_1, \dots, a_n \in \mathbb{Z}$ and $g_1, \dots, g_n \in G$, we have $$e(g_1^{a_1}, \dots, g_n^{a_n}) = e(g_1, \dots, g_n)^{a_1 \dots a_n}$$
and further $e$ is non-degenerate in the sense that $e(g,\dots,g)$ is a generator of $G_T$ for any generator $g$ of $G$.

\subsection{Fully Homomorphic Encryption and Graded Encoding Schemes}

One of the interesting importance of multilinear maps arises in the notion of one of the revolution which swept the world of cryptography, namely fully homomophic encryption (FHE).
The intuition is that FHE ciphertexts behave like the exponents of group elements in a multilinear map, the so called graded encoding scheme \cite{GGH13a}.
Such a scheme is a family of efficient cyclic groups $G_1, \dots, G_n$ of the same prime order $p$ together with efficient non-degenerate bilinear pairings $e: G_i \times G_j \rightarrow G_{i+j}$ whenever $i+j \leq n$. In other words, if we fix a family of generators $g_i$ of the $G_i$'s in such a way that $g_{i+j} = e(g_i, g_j)$, we can add exponents within a given group $G_i$
$${g_i^a}\cdot {g_i^b} = g_i^{a+b};$$
and multiply exponents from two groups $G_i$, $G_j$ as long as $i + j \leq n$: $$e(g_i^a,g_j^b) = g_{i+j}^{a\cdot b}.$$
This makes $g_i^a$ somewhat similar to an FHE encryption of $a$.

\subsection{Generalization of Multilinear Maps to any Group}

Here we generalize the definition of a multilinear map to arbitrary groups $G$ and $G_T$. We say that a map $e:G^n \rightarrow G_T$ is a (symmetric) $n$-linear map (or a multilinear map) if for any $a_1, \dots, a_n \in \mathbb{Z}$ and $g_1, \dots, g_n \in G$, we have $$e(g_1^{a_1}, \dots, g_n^{a_n}) = e(g_1, \dots, g_n)^{a_1 \dots a_n}.$$
Notice that the map $e$ is not necessarily linear in each component.
In addition, we say that $e$ is non-degenerate if there exists $g\in G$ such that $e(g,\dots,g)\neq 1$.

\section{Preliminaries}

\subsection{Nilpotent and Engel Groups}

A group $G$ is said to be nilpotent if it has a finite series 
$$\{1\}=G_{0}< G_{1}<\dots< G_{n}=G$$
which is central, that is, each $G_i$ is normal in $G$ and $G_{i+1}/G_{i}$ is contained in the center of $G/G_{i}$. The length of a shortest central series is the (nilpotency) class of $G$. Of course, nilpotent groups
of class at most 1 are abelian. A great source of nilpotent groups is the class of finite $p$-groups, i.e., finite groups whose orders are powers of a prime $p$.

Close related to nilpotent groups is the calculus of commutators. Let $g_{1},\dots, g_n$ be elements of a group $G$. We will use the following commutator notation: $[g_1, g_2] = g_1^{-1} g_2^{-1} g_1 g_2$.
More generally, a simple commutator of weight $n\geq 2$ is defined recursively by the rule
$$[g_{1},\dots,g_{n}]=
[[g_{1},\ldots,g_{n-1}],g_{n}],$$
where by convention $[g_1]=g_1$. A useful shorthand notation is
$$[x,_n g]=
[x, \underbrace{g,\dots, g}_n].$$
For the reader convenience, we recall the following property of commutators:
\begin{equation}\label{property}
[xy,z]=[x,z]^y[y,z]\ \ {\rm and}\ \ [x,yz]=[x,z][x,y]^z\ {\rm for\ all}\ x,y,z\in G.
\end{equation}
For further basic properties of commutators we refer to \cite[5.1]{Rob}.

It is useful to be able to form commutators of subsets as well as elements. Let $X_1, X_2, \dots$ be nonempty subsets of a group $G$. Define the commutator subgroup of $X_1$ and $X_2$ to be
$$[X_1,X_2]=\langle[x_1,x_2] \,\vert\, x_1\in X_1,x_2\in X_2\rangle.$$
More generally, let
$$[X_{1},\dots,X_{n}]=
[[X_{1},\ldots,X_{n-1}],X_{n}]$$
where $n\geq 2$. Then, there is a natural way of generating a descending sequence of commutator subgroups of a group, by repeatedly commuting with $G$. The result is
a series
$$G=\gamma_1(G)\geq \gamma_2(G)\geq\dots$$
in which $\gamma_{n+1}(G)=[\gamma_n(G),G]$. This is called the lower central series of $G$ and it does not in general reach $1$. Notice that $\gamma_n(G)/ \gamma_{n+1}(G)$ lies in the center of $G/\gamma_{n+1}(G)$.

A useful characterization of nilpotent groups, in terms of commutators, is the following.
 
\begin{lemma}\label{nilpotent}
A group $G$ is nilpotent of class at most $n\geq 1$ if and only if the identity $[g_1,\dots,g_{n+1}]=1$ is satisfied in $G$, that is $\gamma_{n+1}(G)=1$. In particular, in a nilpotent group of class $n$, the subgroup $\gamma_n(G)$ is central.
\end{lemma}
 
Among the best known generalized nilpotent groups are the so-called Engel groups. A group $G$ is called $n$-Engel if $[x,_n y]= 1$ for all $x,y \in G$.
If $G$ is nilpotent of class $n$, then $G$ is $n$-Engel. Also, there are nilpotent groups of class $n$ which are not $(n-1)$-Engel. For example, given a prime $p$, the wreath product $G={\mathbb{Z}}_p \wr {\mathbb{Z}}_p$ is nilpotent of class $p$ but not $(p-1)$-Engel \cite[Theorem 6.2]{Liebeck}. 

Conversely, any finite $n$-Engel group is nilpotent, by a well-known result of Zorn \cite[12.3.4]{Rob}.

\subsection{Nilpotent Group Identities}
 
The next result is a straightforward application of (\ref{property}), together with Lemma \ref{nilpotent}. 

\begin{lemma}\label{1}
Let $G$ be a nilpotent group of class $n>1$ and let $a$ be a nonzero integer. Then, for all $g_1,\dots,g_n\in G$, we have
$$[[g_{1},\ldots,g_{n-1}]^a,g_{n}]=[g_{1},\dots,g_{n}]^a$$
and 
$$[g_{1},\ldots,g_{n-1},g_{n}^a]=[g_{1},\dots,g_{n}]^a.$$
\end{lemma}
 
Then the following proposition holds:
 
\begin{proposition} \label{identity}
Let $G$ be a nilpotent group of class $n>1$. Then
\begin{equation}\label{multi}
[g_1,\dots, g_{i-1}, g_i^{a_i}, g_{i+1},\dots,g_{n}]=[g_1,\dots, g_{i-1}, g_i, g_{i+1},\dots,g_{n}]^{a_i},
\end{equation}
for any $i\in\{1,\dots,n\}$, $a_i\in \mathbb{Z}\backslash\{0\}$ and $g_i\in G$.
\end{proposition}

\begin{proof}
We argue by induction on $n$. The case $n=2$ is true by Lemma \ref{1}.

Let $n>2$. Then $G/\gamma_n(G)$ is nilpotent of class $n-1$. Moreover, $\gamma_n(G)$ is central by Lemma \ref{nilpotent}. Hence the induction hypothesis gives
$$g:=[g_1,\dots, g_{i-1}, g_i^{a_i}, g_{i+1},\dots,g_{n-1}]=[g_1,\dots,g_{n-1}]^{a_i}\ {\rm mod}\ \gamma_n(G).$$
It follows that $g=[g_1,\dots,g_{n-1}]^{a_i}h$ where $h\in\gamma_n(G)$. Since $\gamma_n(G)$ is central, applying (\ref{property}), we get
$$[g,g_n]=[[g_1,\dots,g_{n-1}]^{a_i}h,g_n]=[[g_1,\dots,g_{n-1}]^{a_i},g_n]$$
and so
$$[[g_1,\dots,g_{n-1}]^{a_i},g_n]=[g_1,\dots, g_{i-1}, g_i, g_{i+1},\dots,g_{n}]^{a_i}$$
by Lemma \ref{1}. 
\end{proof}
 
Let $G$ be a nilpotent group of class $n>1$ and $g_1,\dots, g_{n}\in G$. According to Proposition \ref{identity} for any $a_1,\dots,a_{n} \in \mathbb{Z}\backslash\{0\}$, we have
 $$[g_1^{a_1},\dots,g_{n}^{a_{n}}]=[g_1,\dots,g_{n}]^{{\prod}^{n} _{i=1} a_i}.$$
 Therefore we can construct the multilinear map $e:G^{n} \rightarrow G$ given by
$$e({g_1},\dots, {g_{n}})=[g_1,\dots,g_{n}].$$

Similarly, given $x\in G$, we can consider the multilinear map $e':G^{({n-1})} \rightarrow G$ given by
$$e'({g_1},\dots, {g_{n-1}})=[x,{g_1},\dots, {g_{n-1}}].$$
Further, assuming that $G$ is not $(n-1)-$Engel, one can take $x\in G$ in such a way that $e'$ is non-degenerate. In fact there exists $g\in G$ such that $[x,_{n-1} g]\neq 1$.

\section{Multilinear Cryptography using Nilpotent Groups}
Here we propose two multilinear cryptosystems based on the identity (\ref{multi}) in Proposition \ref{identity}.

\subsection{Protocol I}

First we generalize the bilinear map which has been mentioned in \cite{MS}, to multilinear ($n$-linear) map for $n+1$ users. Let ${\cal{A}}_1, \dots, {\cal{A}}_{n+1}$ be the users with private exponents $a_1, \dots, a_{n+1}$ respectively. Given an integer $a\neq 0$, the main formula on which our key-exchange protocol is based on, is an identity in a public nilpotent group $G$ of class $n>1$ (see Proposition \ref{identity}):
$$[{g_1}^a, g_2, \dots, g_n] = [g_1, {g_2}^a, \dots, g_n] = [g_1, g_2, \dots, {g_n}^a]=[g_1, g_2, \dots, g_{n}]^a.$$
The users ${\cal{A}}_j$'s transmit in public channel $${g_i}^{a_j}, \text{ for}\  i = 1, \dots, n; j = 1, \dots, n+1.$$
The key exchange works as follows:
\begin{itemize}
\item The user ${\cal{A}}_1$ can compute $[{g_1}^{a_2}, \dots, {g_{n}}^{a_{n+1}}]^{a_1}$.

\item The user ${\cal{A}}_j$ ($j =2, \dots, n$) can compute $$[{g_1}^{a_1}, \dots, {g_{j-1}}^{a_{j-1}}, {g_j}^{a_{j+1}}, {g_{j+1}}^{a_{j+2}}, \dots, {g_{n}}^{a_{n+1}}]^{a_j}.$$

\item The user ${\cal{A}}_{n+1}$ can compute $[{g_1}^{a_1}, \dots, {g_{n}}^{a_{n}}]^{a_{n+1}}$.
\end{itemize}
The common key is $[g_1,\dots, g_{n}]^{\prod_{j=1} ^{n+1} a_j}$.\\

\noindent {\bf Example: Trilinear Cryptography using Nilpotent Groups of class 3.}
Let $\cal{A}, \cal{B}, \cal{C}, \cal{D}$ be the users with private exponents $a,b,c,d$ respectively. 
The users $\cal{A}$, $\cal{B}$, $\cal{C}$ and $\cal{D}$ transmit in public channel $$x^a, y^a, z^a , x^b, y^b, z^b, x^c, y^c, z^c, x^d, y^d, z^d \text{ respectively.}$$
The key exchange works as follows:
\begin{itemize}
\item The user $\cal{A}$ can compute $[x^b, y^c, z^d]^a$.
\item The user $\cal{B}$ can compute $[x^a, y^c, z^d]^b$.
\item The user $\cal{C}$ can compute $[x^{a}, y^{b}, z^{d}]^{c}$.
\item The user $\cal{D}$ can compute $[x^{a}, y^{b}, z^{c}]^{d}$.
\end{itemize}
The common key is $[x, y, z]^{abcd}$.

\subsection{Protocol II}

Let $G$ be a public nilpotent group of class $n+1$ which is not $n$-Engel ($n\geq 1$). Then there exist $x,g\in G$ such that $[x,_n g]\neq 1$. Suppose that $n+1$ users ${\cal{A}}_1,\dots,{\cal{A}}_{n+1}$ want to agree on a shared secret key.
Each user ${\cal{A}}_j$ selects a private nonzero integer $a_j$, computes $g^{a_j}$ and sends it to the other users. Then:
\begin{itemize}
\item The user ${\cal{A}}_1$ computes $[x^{a_1},g^{a_2},\dots,g^{a_{n+1}}]$.
\item The user ${\cal{A}}_j$ $(j=2\dots,n)$, computes $[x^{a_j}, g^{a_1}, \dots,g^{a_{j-1}},g^{a_{j+1}},\dots, g^{a_{n+1}}]$.
\item The user ${\cal{A}}_{n+1}$ computes $[x^{a_{n+1}},g^{a_1},\dots,g^{a_n}]$.
\end{itemize}
Hence, again by Proposition \ref{identity}, each user obtains $[x,_n g]^{{\prod}_{j=1} ^{n+1} a_j}$ which is the shared key.

\section{Security and Platform Group}
The security of our protocols is based on the discrete logarithm problem (DLP). The ideal platform group for our protocols must be a nonabelian nilpotent group of large order such that the nilpotency class is not too large and the DLP in such a group is hard. Please note that we do not suppose that the group is presented by generating elements and defining relators or as a subgroup of a triangular matrix group over a prime finite field (in finite case) or over the ring of integers (in torsion-free case).

In \cite{AVS2011}, Sutherland has studied the DLP in finite abelian $p$-groups, and showed how to apply the algorithms for $p$-groups to find the structure of any finite abelian group.

In a series of papers by Mahalanobis, 
the DLP has been studied for finite $p$-groups but mostly for nilpotent groups of class $2$ \cite{AMIJM,AM2015}. In particular, in \cite{MS}, Mahalanobis and Shinde proposed $p$-groups of class $2$ in which the platform is not practical as showed by the authors.

\vspace{0.5cm}

\noindent
{\bf Funding}.  The authors were supported by the ``National Group for Algebraic and Geometric Structures, and their Applications'' (GNSAGA -- INdAM). The first author was also partially supported by a PSC-CUNY grant from the CUNY Research Foundation.

\end{document}